\documentclass[11pt]{article}

\usepackage{amsthm}
\usepackage{amsfonts}
\usepackage{amscd}
\usepackage{amssymb}
\usepackage{textcomp}
\usepackage{color}
\usepackage{amsmath}
\usepackage{amsfonts}
\usepackage{graphicx}
\usepackage{mathtools}
\usepackage{times}
\usepackage{bm}
\usepackage{natbib}
\usepackage{url}
\usepackage{thmtools}
\usepackage{tikz-cd}
\usepackage{adjustbox}

\usepackage{geometry}
\newcommand{\R}{\mathbb{R}}

\newcommand{\T}{\mathcal{T}}

\newcommand{\G}{\mathcal{G}}
\newcommand{\tauenv}{\hat{\tau}_{\text{env}}}
\newcommand{\tauD}{\hat{\tau}_{\text{1D}}}

\newcommand{\upenv}{\hat{\upsilon}_{\text{env}}}
\newcommand{\upD}{\hat{\upsilon}_{\text{1D}}}
\newcommand{\Menv}{M_{\text{env}}}
\newcommand{\MD}{M_{\text{1D}}}
\newcommand{\genv}{\hat{g}_{\text{env}}}
\newcommand{\inner}[1]{\langle #1 \rangle}

\DeclareMathOperator{\Var}{Var}

\declaretheoremstyle[
  spaceabove=\topsep, spacebelow=\topsep,
  headfont=\itshape\scshape,
  notefont=\mdseries, notebraces={(}{)},
  bodyfont=\itshape,
  postheadspace=1em,
  qed={}
]{mythmstyle}

\newtheorem{thm}{Theorem} 
\newtheorem{prop}{Proposition} 
 
\newtheorem{defn}{Definition}
\allowdisplaybreaks

\begin{document}




\begin{small}
\begin{center}

\noindent {\LARGE \bf Combining Envelope Methodology and Aster Models for Variance Reduction in Life History Analyses}

\vspace{.1in}

Daniel J. Eck 

\vspace{.01in}

\textit{Department of Biostatistics, Yale School of Public Health.} \\
\textit{daniel.eck@yale.edu}

\vspace{.1in}

Charles J. Geyer \qquad and \qquad R. Dennis Cook

\vspace{.01in}

\textit{Department of Statistics, University of Minnesota}

\end{center}
\end{small}



\begin{abstract}
Precise estimation of expected Darwinian fitness, the expected lifetime number 
of offspring of organism, is a central component of life history analysis. 
The aster model serves as a defensible statistical model for distributions of 
Darwinian fitness. The aster model is equipped to incorporate the major life 
stages an organism travels through which separately may effect Darwinian fitness. 
Envelope methodology reduces asymptotic variability by establishing 
a link between unknown parameters of interest and the asymptotic covariance 
matrices of their estimators. It is known both theoretically and in 
applications that incorporation of envelope methodology reduces asymptotic 
variability. We develop an envelope framework, including a new envelope 
estimator, that is appropriate for aster analyses.
The level of precision provided from our methods allows researchers to draw 
stronger conclusions about the driving forces of Darwinian fitness from their 
life history analyses than they could with the aster model alone.  
Our methods are illustrated on a simulated dataset and a life history analysis 
of \emph{Mimulus guttatus} flowers is provided. Useful variance 
reduction is obtained in both analyses.
\end{abstract}

Darwinian fitness; fitness landscape; envelope model; parametric bootstrap


\section{Introduction}
\label{sec:intro}

The estimation of expected Darwinian fitness 
is a very important procedure in both biology and genetics. 
The importance of this is not just limited to scientific disciplines, it is 
important for public policy. With genetic theory and simulation studies, 
\cite{burger} shows that, under certain conditions, a changing environment 
leads to extinction of species. In a field study, \cite{etterson} argued that 
the predicted evolutionary response to predicted rates of climate change are 
far too slow. In these papers, and all life history analyses of their kind, 
expected Darwinian fitness is the response variable. The interesting 
scientific conclusions are drawn from it. 

In many life history analyses, values of expected Darwinian fitness are 
plotted using a fitness landscape \citep{lande, shaw}. A fitness landscape is 
the conditional expectation of Darwinian fitness given phenotypic trait values 
considered as a function of those values. When fitness is the response 
variable in a regression model and phenotypic traits are the covariates, the 
fitness landscape is the regression function. Estimation of the fitness 
landscape began with \cite{lande}. 
They use ordinary least squares regression of fitness on phenotypes to 
estimate the best linear approximation of the fitness landscape and quadratic 
regression to estimate the best quadratic approximation. Here ``best'' means 
minimum variance unbiased, as in the Gauss-Markov theorem. Their use of $t$ 
and $F$ tests and confidence intervals requires the assumption that fitness 
is conditionally homoscedastically normally distributed given phenotypic 
trait values.  This assumption is almost always grossly incorrect when one 
uses a good surrogate for Darwinian fitness \citep{Olds, shaw2}.

Aster models \citep{geyer} were designed to fix all of the problems of the 
\cite{lande} approach and of all other approaches to life history analysis 
\citep{shaw2}. The aster model is the state-of-the-art model for all 
life history analyses in which the estimation of expected Darwinian fitness 
is the primary goal. \citet{geyer, shaw2, stanton, shaw2015, eck} show 
various kinds of life history data for which aster models are necessary. 
Assumptions for aster models are given in Section 2 below.

In this article we combine envelope methodology \citep{cook, cook2, su} with 
aster models \citep{geyer, shaw2} to estimate the fitness landscape 
\citep{lande, shaw} in life history analysis. The primary emphasis is that 
this combination of methods estimates the fitness landscape with less 
variability than is possible with aster models alone.
We first show how existing envelope estimators constructed from the 1D 
algorithm \citep{cook2, cook3, zhangmai} can reduce variability in estimation 
of the fitness landscape. We then develop a new envelope estimator that 
avoids the potential numerical pitfalls of the 1D algorithm.
Variance reduction is assessed using parametric bootstrap techniques 
in \citet[Section 4]{efron}. These bootstrap algorithms account for 
variability in model selection. Our methodology provides the most precise 
estimation of expected Darwinian fitness to date. 
Researchers using our methods can therefore draw stronger conclusions about 
the driving forces of Darwinian fitness from their life history analyses. 

In a life history analysis of \emph{M. guttatus} flowers and a simulated 
example, we show that our methodology leads to variance reduction in 
estimation of expected Darwinian fitness when compared with analyses that 
use aster models alone. We show that this variance reduction leads to sharper 
scientific inferences about the potential causes of Darwinian fitness in the 
\emph{M. guttatus} life history analysis. 
Our examples are fully reproducible, and the calculations necessary for their 
reproduction are included in an accompanying technical report 
\citep{ecktech}.

\section{The aster model}
\label{s:model}

Aster models are regular full exponential families.  Parameters are 
estimated by maximum likelihood.  If $Y$ is the response vector, 
$\mu = E(Y)$ is the saturated model mean value parameter, and $M$ 
is the model matrix for an unconditional canonical affine submodel, 
then the maximum likelihood estimate of $\mu$ satisfies 
$M^T \hat{\mu} = M^T y$ \citep[Section 3.2]{geyer}. 
$\tau = M^T \mu$ is the submodel mean value parameter 
\citep[Section~2.4]{geyer}.  Likelihood ratio tests 
for model comparison and confidence intervals for all parameters 
are based on the usual asymptotics and Fisher information are provided 
by R package \texttt{aster} \citep{aster-package}.  In particular, for this 
article we need to know that $\hat{\tau} = M^T y$ is is a minimum 
variance unbiased estimator of the parameter it estimates and its 
exact variance matrix is the Fisher information matrix for the 
submodel canonical parameter $\beta$ (the vector of regression 
coefficients).  That is, in the usual asymptotics of maximum likelihood 
for this parameter the mean and variance are exact not approximate; only 
the normal distribution is approximate. 

The aster model is a directed acyclic graphical model 
\citep[Section 3.2.3]{lauritzen} in which the joint density is 
a product of conditional densities that are specified by the arrows depicted 
in the graph. Lines that appear in the graph specify nodes which are 
dependent. For example, an organism may have multiple paths in their life 
history and can only go down one of them \citep{eck}. 
The aster model follows five assumptions which are:  
A1 The graph of arrows is acyclic.
A2 In the graph of lines every connected component is a complete graph,
   which is called a \emph{dependence group}.
A3 Every node in a dependence group with more than one node has the same
   predecessor (there is an arrow from the predecessor to each
   node in the group).  Every dependence group consisting of exactly
   one node has at most one predecessor.
A4 The joint distribution is the product of conditional distributions,
   one conditional distribution for each dependence group.
A5 Predecessor is sample size, meaning each conditional distribution
   is the distribution of the sum of $N$ independent and identically
   distributed random vectors, where $N$ is the value of the
   predecessor, the sum of zero terms being zero.
A6 The conditional distributions are exponential families having the
   components of the response vector for the dependence group as
   their canonical statistics.
Assumptions A5 and A6 mean for an arrow $y_k \rightarrow y_j$
that $y_j$ is the sum of independent and identically distributed random
variables from the exponential family for the arrow and there are $y_k$
terms in the sum (the sum of zero terms is zero). These assumptions
imply that the joint distribution of the aster model is an exponential family
\citep[Section 2.3]{geyer}. Three of these assumptions have a clear 
biological meaning as well. Assumptions A1 through A3 restricts an individual 
from revisiting life stages that have come to pass. Assumption A5 implies that 
dead individuals remain dead and have no offspring though the course of the 
study. 


These aster models are saturated, having one parameter per component
of the response vector, and are not useful. Hence, as with linear and 
generalized linear models, canonical affine submodels are used with change 
of parameter $\varphi = a + M \beta$, where $\varphi$ is the saturated model 
parameter vector (linear predictor in the terminology of generalized linear 
models), $\beta$ is the submodel parameter (``coefficients'' in R 
terminology), $a$ is the offset vector, and $M$ is the model matrix.
The aster submodel has log likelihood
$
  l(\beta) = \langle M^TY, \beta \rangle - c(a + M\beta)
$
where $Y$ is the response vector and $c(\cdot)$ is the cumulant function of 
the exponential family. 

There are three parameters of interest that are present in the aster analyses
we consider (see \cite{geyer3} and \cite{ecktech} for more detail on aster model 
parameterizations). These parameterizations are 
1) the aster submodel canonical parameter vector $\beta \in \R^p$,
2) the aster submodel mean-value parameter vector $\tau \in \R^p$,
3) the saturated aster model mean-value parameter vector $\mu \in \R^m$, where 
$m$ is the number of individuals sampled multiplied by the number of nodes in 
the aster graph. 
These three parameterizations are all linked via invertible 1-1 
transformations when the model matrix $M$ is of full rank. 
The usual asymptotics of maximum likelihood estimation give  
\begin{equation}\label{asymptotics-MLE-mv}
  \surd{n}\left(\hat{\tau} - \tau \right) \overset{d}{\longrightarrow} 
    N\left(0,\, \Sigma\right),
\end{equation}
where $\Sigma = \Var(M^T Y)$ is the Fisher information matrix associated with 
the canonical parameter vector $\beta$. The maximum likelihood estimator of 
$\beta$ is asymptotically normal with variance given by $\Sigma^{-1}$. 
From \eqref{asymptotics-MLE-mv} and the delta method we can obtain the 
asymptotic distribution for any differentiable function of $\hat{\tau}$. 
The asymptotic distribution for a differentiable function $g$ of $\hat{\tau}$ 
is
\begin{equation}\label{Delta-method-from-mv}
  \surd{n}\left\{g(\hat{\tau}) - g(\tau)\right\} \overset{d}{\longrightarrow} 
    N\left\{0,\; \nabla g(\tau)\Sigma \nabla g(\tau)^T\right\}.
\end{equation}
In particular, the asymptotic distribution of estimated expected Darwinian 
fitness is of interest. Let $h(\mu)$ be expected Darwinian fitness.  
Both $\beta$ being a function of $\tau$ and $\mu = \nabla c(a + M \beta)$ 
imply that 
$
  g(\tau) =  h\Bigl[\nabla c\bigl\{a + M f(\tau) \bigr\}\Bigr]
$
is expected Darwinian fitness as a function of $\tau$ and is differentiable 
if $h$ is differentiable. 
The estimator $g(\hat{\tau})$ has asymptotic distribution given by 
\eqref{Delta-method-from-mv}. We have the potential to do better through the 
incorporation of envelope methodology.

\section{Incorporation of envelope methodology}
\label{s:envelopemodel}

Aster model estimates of expected Darwinian fitness may be too variable to 
be useful, and in consequence we may not be able to statistically distinguish 
estimates of expected Darwinian fitness over the 
fitness landscape.
We address this problem through the incorporation of envelope 
methodology into the aster modelling framework. 

Envelope models were developed originally as a variance reduction tool for 
the multivariate linear regression model. In this article we focus on envelope 
methodology for general vector-valued parameter estimation \citep{cook2}.
Envelope methodology has the potential to reduce the 
variability of any $\surd{n}$-consistent asymptotically normal distributed 
consistent estimator. 
We now define the envelope subspace.
Let $\upsilon$ be a parameter of interest and suppose that $\hat{\upsilon}$ is 
a $\surd{n}$-consistent estimator of $\upsilon$ with asymptotic covariance 
matrix $\Sigma_{\upsilon,\upsilon}$. Let 
$
  \T = \text{span}(\upsilon) =\{a\upsilon : a \in \R\}. 
$ 
The envelope subspace $\mathcal{E}_{\Sigma_{\upsilon,\upsilon}}(\T)$ is 
defined as the intersection of all reducing subspaces of 
$\Sigma_{\upsilon,\upsilon}$ that contain $T$ (a reducing subspace is a sum of 
eigenspaces if all eigenvalues of $\Sigma$ have multiplicity one). 
The envelope space satisfies both
$$
  \T \subset \mathcal{E}_{\Sigma_{\upsilon,\upsilon}}(\T),
    \qquad
  \Sigma_{\upsilon,\upsilon} 
    = P_{\mathcal{E}}\Sigma_{\upsilon,\upsilon}P_{\mathcal{E}} 
    + Q_{\mathcal{E}}\Sigma_{\upsilon,\upsilon}Q_{\mathcal{E}};
$$
where $P_{\mathcal{E}}$ is the projection into the envelope subspace and 
$Q_{\mathcal{E}}$ is the projection into the orthogonal complement. In 
coordinate form, these two envelope conditions are 
$$
  \T \subset \text{span}(\Gamma),
    \qquad 
  \Sigma_{\upsilon,\upsilon} 
    = \Gamma\Omega\Gamma^T + \Gamma_o\Omega_o\Gamma_o^T;
$$
where $(\Gamma,\Gamma_o)$ is a partitioned orthogonal matrix, the columns of 
$\Gamma$ are a basis for $\mathcal{E}_{\Sigma_{\upsilon,\upsilon}}(\T)$, and 
the dimensions of the positive definite matrices $\Omega$ and $\Omega_o$ are 
such that the matrix multiplications are defined. The quantities 
$P_{\mathcal{E}}\Sigma_{\upsilon,\upsilon}P_{\mathcal{E}}$ and 
$\Gamma\Omega\Gamma^T$ are often referred to as 'material information' 
in the envelope literature \citep{cook,cook2} since they represent the 
portion of variability that is necessary for the task of estimating 
$\upsilon$. Similarly, 
$Q_{\mathcal{E}}\Sigma_{\upsilon,\upsilon}Q_{\mathcal{E}}$
and $\Gamma_o\Omega_o\Gamma_o^T$ are referred to as 'immaterial information' 
since they represent extraneous variability.

Intuitively, the envelope estimator reduces variability in estimation at no 
cost to consistency. An illuminating depiction and explanation of how an 
envelope analysis increases efficiency in multivariate linear regression 
problems was given by \citet[pgs. 134--135]{su}. The same intuition applies 
to envelope methodology of \citet{cook2}. 
In applications, there is a cost to estimate 
$u = \text{dim}\left\{\mathcal{E}_{\Sigma_{\upsilon,\upsilon}}(\T)\right\}$ 
and $\Gamma$. With the basis matrix $\Gamma$ estimated, we can then assess the 
variance reduction of the envelope estimator through the parametric bootstrap. 

The 1D algorithm \citep[Algorithm 2]{cook2} estimates a basis matrix 
$\Gamma$ for $\mathcal{E}_{\Sigma_{\upsilon,\upsilon}}(\T)$ at a 
user-supplied envelope dimension $u$. The estimate of $\Gamma$ is obtained by 
providing $\widehat{\Sigma}_{\upsilon,\upsilon}$ and 
$\hat{\upsilon}\hat{\upsilon}^T$ as inputs into the 1D algorithm. The 
resulting estimator of $\Gamma$ obtained from the 1D algorithm, 
$\widehat{\Gamma}_u$, gives a $\surd{n}$ consistent estimator 
$P_{\hat{\mathcal{E}}}$ of the projection onto the envelope subspace 
$P_{\mathcal{E}}$ \citep{cook2}. 
The 1D algorithm can be used to estimate $u$ consistently 
\citep[Theorem 2]{zhangmai}.

Define $\upD = P_{\hat{\mathcal{E}}}\hat{\upsilon}$ to be the envelope 
estimator of $\upsilon$ where $\upsilon$ are parameters that link the 
estimation of expected Darwinian fitness to covariates of interest. We write 
$\tau = (\gamma^T$, $\upsilon^T)^T$ where $\gamma$ are aster model 
parameters not linking covariates to the estimation of Darwinian fitness. 
The envelope estimator of $\tau$ is given as 
$$ 
  \tauD = \left(\begin{array}{c}
    \widehat{\gamma} \\ 
    \upD\end{array}\right) = \left(\begin{array}{cc}
      I & 0 \\
      0 & P_{\hat{\mathcal{E}}}\end{array}\right) M^T Y = \MD^T Y,
  \qquad
  \MD = M\left(\begin{array}{cc}
    I & 0 \\
    0 & P_{\hat{\mathcal{E}}}\end{array}\right).
$$
The model matrix $\MD$ corresponds to the aster model that incorporates 
the envelope structure.

\begin{prop} \label{Prop1D}
The envelope estimator $\tauD$ is a maximum likelihood estimator of $\tau$ for 
the aster model with model matrix $\MD$.
\end{prop}

\begin{proof}
We have
$
   l_{\text{env}}(\beta) = \inner{Y, \MD\beta} - c(\MD\beta)
     = \inner{\MD^TY, \beta} - c(\MD\beta).
$
and
$
  \nabla_{\beta} l(\beta) = \MD^TY - \nabla_{\beta} c(\MD\beta).
$
Setting $\nabla_{\beta} l(\beta) = 0$ and solving for $\beta$ yields 
$\nabla_{\beta} c(\MD\beta)|_{\beta = \hat{\beta}} = \MD^TY = \tauD$. 
\end{proof}

This proposition justifies the use of the transformations 
to switch between maximum likelihood estimators of the different 
aster model parameterizations. 
We compare envelope dimensions $u$ by transforming envelope estimators of 
$\tau$ to envelope estimators of $\beta$ and then evaluate the 
log likelihood at the envelope estimator of $\beta$.
The randomness inherent in $\MD$ is non-problematic in most applications. 
This is because the 1D algorithm provides a $\surd{n}$-consistent envelope 
estimator of $\upsilon$ even when $u$ is estimated. 
Inferences about aster model parameters implicitly assumes that $n$ is large 
enough for the asymptotic normality to be a good approximation for the 
distribution of maximum likelihood estimators.

We then compute an envelope estimator of expected Darwinian fitness 
$g(\tau)$ using the aster model with model matrix $\MD$. 
Note that the model matrix $\MD$ is not of full column rank. 
Therefore the transformations 
used to switch between aster model parameterizations 
are not 1-1. In particular, many distinct estimates of $\beta$ map to 
$\tauD$. Each of these distinct estimated values of $\beta$ maps to the 
same estimate of $\MD\beta$, which in turn maps to a 
common estimate of expected Darwinian fitness. 
The loss of 1-1 transformations is not an issue in this case.




Our estimator of estimated expected Darwinian fitness 
is given by \citet[equation (4.5)]{efron}, with $g(\cdot)$ replacing 
$t(\cdot)$. Steps for this algorithm are given in Algorithm 1. 
When the top-level of our bootstrap procedure (Steps 1 through 4 in 
Algorithm 1) has run for $B$ iterations, we obtain the envelope estimator 
\begin{equation} \label{envest}
  \hat{g}_{\text{1D}} 
    = \frac{1}{B}\sum_{b=1}^B g\left\{\tauD^{(b)}\right\}.
\end{equation}
The envelope estimator $\hat{g}_{\text{1D}}$ implicitly behaves 
as a weighted average with the weights reflecting the likelihood of observing 
a particular estimated value of $u$, $\upD$, and $\MD$. The value of $u$ is 
estimated with either the Bayesian information criterion or Akaike information 
criterion at every iteration of the parametric bootstrap. The intuition is 
that the averaging in $\hat{g}_{\text{1D}}$ will smooth out variability due 
to the estimating $u$ with our chosen model selection criterion. 
Variability of $\hat{g}_{\text{1D}}$ is estimated using the double 
bootstrap technique in \citet[Section 4]{efron}, with steps shown below in 
Algorithm 1. This bootstrap technique 
accounts for all estimation error, including model selection error. 
The reason it accounts for all estimation error, is that all estimation, 
including model selection is done in each iteration of the bootstrap 
(nothing estimated is ever treated as known in bootstrap iterations). 

\noindent Algorithm 1. Parametric bootstrap for assessing the variability of 
the envelope estimator $\hat{g}_{\text{1D}}$:
\begin{enumerate}
\item[1.] Fit the aster model to the data and obtain $\hat{\upsilon}$ 
  and $\widehat{\Sigma}_{\upsilon,\upsilon}$ from the aster model fit. 
\item[2.] Choose a model selection criterion. Compute the envelope 
  estimator of $\upsilon$ in the original sample, given as 
  $\upD = P_{\hat{\mathcal{E}}}\hat{\upsilon}$ where 
  $P_{\hat{\mathcal{E}}}$ is obtained from the 1D algorithm and the chosen 
  model selection criterion. 
\item[3.] Perform a parametric bootstrap by generating samples from the 
  distribution of the aster submodel evaluated at 
  $\tauD = (\hat{\gamma}^T,\upD^T)^T$. For $b=1$, $\ldots$, $B$ 
  of the procedure: 
  \begin{enumerate}
  \item[(3a)] Compute $\hat{\tau}^{(b)}$ and 
    $\widehat{\Sigma}_{\upsilon,\upsilon}^{(b)}$ from the aster model fit to 
    the resampled data.
  \item[(3b)] Compute $P_{\hat{\mathcal{E}}}^{(b)}$ as done in Step 2. 
  \item[(3c)] Compute 
    $\tauD^{(b)} = \left\{\hat{\gamma}^{(b)^T},\upD^{(b)^T}\right\}^T$  
    and $g\left\{\tauD^{(b)}\right\}$.
  \end{enumerate}
\item[4.] The bootstrap estimator of expected Darwinian 
  fitness is the average of the envelope estimators computed in Step 3c. 
  This completes the first part of the bootstrap procedure. 
\item[5.] 
  At $k = 1$, $\ldots$, $K$, for each $b=1$, $\ldots$, $B$ we:
  \begin{enumerate}
  \item[(5a)] Generate data from the distribution of the aster submodel 
    evaluated at $\tauD^{(b)}$. 
  \item[(5b)] Perform Steps 3a through 3c with respect to the dataset 
    obtained in Step 5a to calculate both 
  $\tauD^{(b)^{(k)}}$ and $g\left\{\tauD^{(b)^{(k)}}\right\}$. 
  \end{enumerate}
\item[6.] Compute both $\hat{g}_{\text{1D}}$ and the standard deviation 
  in \citet[equation (4.15)]{efron}.  
\end{enumerate}

\section{A direct envelope estimator using reducing subspaces} 
\label{s:alternative}

We propose a new way of constructing envelope estimators provided that the 
eigenvalues of $\Sigma_{\upsilon,\upsilon}$ have multiplicity one. In this 
section, envelope estimators are constructed directly from the reducing 
subspaces of $\widehat{\Sigma}_{\upsilon,\upsilon}$. 
Uniqueness of the eigenvalues of $\Sigma_{\upsilon, \upsilon}$ 
implies that its reducing subspaces are sums of its eigenspaces.  
Let $\G$ be a reducing subspace of $\Sigma_{\upsilon,\upsilon}$. 
Define $\Gamma_{\G}$ and $P_{\G}$ as the basis matrix for $\G$ 
and the projection onto $\G$ respectively. 
Let $\widehat{\G}$ be the estimator of the reducing subspace $\G$, 
obtained by searching over the 1 dimensional eigenspaces of 
$\widehat{\Sigma}_{\upsilon,\upsilon}$.  
In applications, eigenvalues of $\widehat{\Sigma}_{\upsilon,\upsilon}$ 
are almost always unique. Define $\widehat{\Gamma}_{\widehat{\G}}$ and 
$
  \widehat{P}_{\widehat{\G}} = 
    \widehat{\Gamma}_{\widehat{\G}}\widehat{\Gamma}^T_{\widehat{\G}}
$ 
as estimators of $\Gamma_{\G}$, and $P_{\G}$ respectively. The basis matrix 
$\widehat{\Gamma}_{\widehat{\G}}$ is constructed from the eigenvectors of 
$\widehat{\Sigma}_{\upsilon,\upsilon}$.
We now define the envelope estimator of $\upsilon$ constructed from 
reducing subspaces.

\begin{defn}
The envelope estimator of $\upsilon$ constructed from the reducing 
subspaces $\G$ is defined to be 
$\upenv = \widehat{P}_{\widehat{\G}}\hat{\upsilon}$. 
\end{defn}

The reducing subspaces of $\widehat{\Sigma}_{\upsilon,\upsilon}$ are 
$\surd{n}$ consistent estimators of the reducing subspaces of 
$\Sigma_{\upsilon,\upsilon}$. Therefore $\widehat{\Gamma}_{\widehat{\G}}$, 
$\widehat{P}_{\widehat{\G}}$, and the corresponding estimator 
$\widehat{P}_{\widehat{\G}}\hat{\upsilon}$ are $\surd{n}$ consistent 
estimators of $\Gamma_{\G}$, $P_{\G}$, and $\upsilon$ respectively. 
The envelope estimator of $\tau$ constructed by reducing subspaces is given by 
$$
  \tauenv = \left(\begin{array}{c}
    \widehat{\gamma} \\ 
    \upenv\end{array}\right) = \left(\begin{array}{cc}
      I & 0 \\
      0 & \widehat{P}_{\widehat{\G}}\end{array}\right) M^T Y = \Menv^T Y,
  \qquad
  \Menv = M\left(\begin{array}{cc}
    I & 0 \\
    0 & \widehat{P}_{\widehat{\G}}\end{array}\right).
$$
The model matrix $\Menv$ corresponds to the aster model that incorporates 
the envelope structure obtained from the reducing subspaces of 
$\widehat{\Sigma}_{\upsilon,\upsilon}$. We have a similar result as 
Proposition~\ref{Prop1D} for the aster model with model matrix $\Menv$.

\begin{prop} \label{Propred}
The envelope estimator $\tauenv$ is a maximum likelihood estimator of 
$\tau$ for the aster model with model matrix $\Menv$.
\end{prop}


The proof of Proposition~\ref{Propred} follows the same steps as the proof 
for Proposition~\ref{Prop1D}. 
There is a close connection between envelope estimation using reducing 
subspaces and envelope estimation using the 1D algorithm. In the population, 
$\upenv = \upD$.
The connection between both estimation methods exists in finite samples as 
seen in Theorem~\ref{connection}. In preparation, define orthogonal matrices 
$\widehat{O}_u = \left(\widehat{\Gamma}_u,\widehat{\Gamma}_{uo}\right)$,  
$
  \widehat{O}_{\widehat{\G}} 
    = \left(\widehat{\Gamma}_{\widehat{\G}},
      \widehat{\Gamma}_{\widehat{\G}o}\right), 
$ 
and $\widehat{O} = \widehat{O}_{\widehat{\G}}\widehat{O}_u^T$. 
The matrices $\widehat{\Gamma}_{\widehat{\G}}$, 
$\widehat{\Gamma}_{\G o}$, $\widehat{\Gamma}_{u}$, and $\widehat{\Gamma}_{uo}$ 
converges in probability to $\Gamma_{\G}$, $\Gamma_{\G o}$, $\Gamma_{u}$, 
and $\Gamma_{uo}$ respectively. Now 
$\Gamma_{\G}^T\Gamma_{u}$ and $\Gamma_{\G o}^T\Gamma_{uo}$ are both $0$-$1$ 
valued rotation matrices and 
$\Gamma_{\G}^T\Gamma_{uo} = \Gamma_{\G o}^T\Gamma_{u} = 0$. These facts 
imply that $\widehat{O}$ converges in probability to a $0$-$1$ valued rotation  
matrix. We will assume that $\widehat{O} \overset{p}{\to} I$ without loss 
of generality.

\begin{thm} \label{connection}
The basis matrix $\widehat{\Gamma}_{\widehat{\G}}$ is the output of the 1D 
algorithm with inputs 
$\widehat{M} = \widehat{O}\widehat{\Sigma}_{\upsilon,\upsilon}\widehat{O}^T$ 
and 
$\widehat{U} = \widehat{O}\hat{\upsilon}\hat{\upsilon}^T\widehat{O}^T$ 
at dimension $u$. 
\end{thm}

\begin{proof}
Let $\widehat{M}_2 = \widehat{\Sigma}_{\upsilon,\upsilon}$ and 
$\widehat{U}_2 = \hat{\upsilon}\hat{\upsilon}^T$.
Similar to the proof of \cite[Proposition 6]{cook3}, let 
\begin{align*}
  Q_n(g) &= -n/2\log(g^T\widehat{O}\widehat{M}_2\widehat{O}^Tg) 
    - n/2\log\left[g^T\left\{\widehat{O}\left(\widehat{M}_2 + \widehat{U}_2\right)\widehat{O}^T\right\}^{-1}g\right]
    + n \log(g^T g) \\
  &= -n/2\log(g^T\widehat{O}\widehat{M}_2\widehat{O}^Tg) 
    - n/2\log\left\{g^T\widehat{O}\left(\widehat{M}_2 + \widehat{U}_2\right)^{-1}\widehat{O}^Tg\right\}
    + n \log(g^T\widehat{O}\widehat{O}^T g).  
\end{align*}
Now let $\hat{v}_k$, $k = 1$, $\ldots$, $u$ be the $k$th column of 
$\widehat{\Gamma}_{\G}$ and define $A \in \R^{p\times p}$ to be a matrix of 
$0$'s with $1$'s occupying the first $k$ diagonal entries. Then 
\begin{align*}
  Q_n(\hat{v}_k) 
  &= -n/2\log(\hat{v}_k^T\widehat{O}\widehat{M}_2\widehat{O}^T\hat{v}_k) 
    - n/2\log\left\{\hat{v}_k^T\widehat{O}\left(\widehat{M}_2 
       + \widehat{U}_2\right)^{-1}\widehat{O}^T\hat{v}_k\right\} \\
    &\qquad+ n \log(\hat{v}_k^T\widehat{O}^T\widehat{O}\hat{v}_k) \\
  &= - n/2\log(A_k \widehat{O}_u\widehat{M}_2\widehat{O}_u^TA_k) 
     - n/2\log\left\{A_k\widehat{O}_u\left(\widehat{M}_2 
       + \widehat{U}_2\right)^{-1}\widehat{O}_u^TA_k\right\} \\
    &\qquad+ n \log(A_k\widehat{O}_u^T\widehat{O}_u A_k) \\
  &= - n/2\log(\hat{g}_{uk}^T\widehat{M}_2\hat{g}_{uk}) 
     - n/2\log\left\{\hat{g}_{uk}^T\left(\widehat{M}_2 
       + \widehat{U}_2\right)^{-1}\hat{g}_{uk}\right\} \\
    &\qquad+ n \log(\hat{g}_{uk}^T\hat{g}_{uk})
\end{align*}
where $\hat{g}_{uk}$ is the $k$th column of $\widehat{\Gamma}_u$, 
the output of the 1D algorithm with $\widehat{M}_2$ 
and $\widehat{U}_2$ as inputs. Therefore $\hat{v}_k$ is a maximizer 
of $Q_n(g)$ and this completes the proof.
\end{proof}

Theorem~\ref{connection} in combination with \citet[Theorem 2]{zhangmai} 
allows for us to estimate $\G$ consistently. 
Thus the variability associated with the estimation of $\upenv$ and 
$\G$ decreases as $n\to\infty$. However in practical applications correct 
model selection cannot be guaranteed.
Therefore the envelope estimator of expected Darwinian fitness 
$g(\tauenv)$ has an extra source of variability due to model selection 
uncertainty. 
We develop a double bootstrap procedure with steps similar to those in 
Algorithm 1 to account for variability in model selection. 
The first level of the bootstrap procedure provides the 
estimator of expected Darwinian fitness,  
\begin{equation} \label{genv}
  \hat{g}_{\text{env}} 
    = \frac{1}{B}\sum_{b=1}^B g\left\{\tauenv^{(b)}\right\}.
\end{equation}
The same model selection criteria is used to select the reducing subspace used 
to construct $\tauenv^{(b)}$ at every iteration $b = 1$, $\ldots$, $B$.
The second of level of this bootstrap procedure estimates the variability 
of \eqref{genv}. The steps for this algorithm are provided in \citet{ecktech}. 
The utility of our double bootstrap procedure is shown in Section 5.1 where, in 
that example, there is considerable disagreement between model selection 
criteria. 

When $k$ is small, $\upenv$ is preferable to $\upD$.
At any iteration of the 1D algorithm, minimizers of the objective function 
stated in \cite[Algorithm 2]{cook2} are pulled towards reducing subspaces of 
$\widehat{\Sigma}_{\upsilon,\upsilon}$. This objective function is non-convex 
and contains potentially many local minima. The optimizations conducted within 
the 1D algorithm are sensitive to starting values and can get stuck at these 
local minima. This undermines the 1D algorithm since it is required that users 
find global minima for its justification. Unlike the 1D algorithm, the 
reducing subspace approach does not involve any optimization routines. 
However when $k$ is moderately large, the computation 
of all of candidate envelope estimators at each reducing subspace is too 
computationally intensive, there are $2^k - 1$ possible reducing subspaces in 
non trivial problems. The 1D algorithm may still be fast when $k$ is 
moderately large, only $k-1$ optimizations in non trivial problems.
In most aster applications $k$ is small since data is obtained through 
expensive collection methods and fitness landscapes are 
low-dimensional \citep{shaw, eck}.

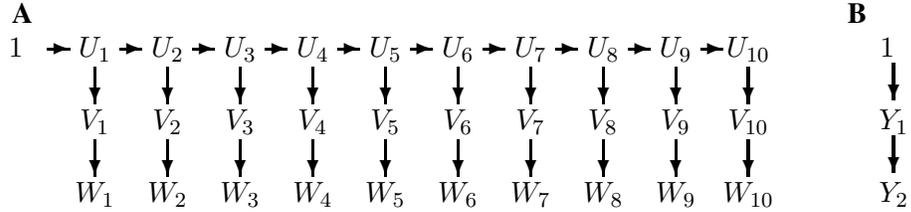
\begin{figure*}[t]
\begin{center}
    \centering
    \vspace{0pt}
    \setlength{\unitlength}{0.38 in}
    \thicklines
    \begin{picture}(10.3,2.8)(-0.1,-2.3)
      \put(0, 0.5){\makebox(0,0){\textbf{A}}}
      \put(11.5, 0.5){\makebox(0,0){\textbf{B}}}
      \put(0,0){\makebox(0,0){$1_{\hphantom{0}}$}}
      \put(1,0){\makebox(0,0){$U_1$}}
      \put(2,0){\makebox(0,0){$U_2$}}
      \put(3,0){\makebox(0,0){$U_3$}}
      \put(4,0){\makebox(0,0){$U_4$}}
      \put(5,0){\makebox(0,0){$U_5$}}
      \put(6,0){\makebox(0,0){$U_6$}}
      \put(7,0){\makebox(0,0){$U_7$}}
      \put(8,0){\makebox(0,0){$U_8$}}
      \put(9,0){\makebox(0,0){$U_9$}}
      \put(10,0){\makebox(0,0){$U_{10}$}}
      \multiput(0.35,0)(1,0){10}{\vector(1,0){0.3}}
      \put(1,-1){\makebox(0,0){$V_1$}}
      \put(2,-1){\makebox(0,0){$V_2$}}
      \put(3,-1){\makebox(0,0){$V_3$}}
      \put(4,-1){\makebox(0,0){$V_4$}}
      \put(5,-1){\makebox(0,0){$V_5$}}
      \put(6,-1){\makebox(0,0){$V_6$}}
      \put(7,-1){\makebox(0,0){$V_7$}}
      \put(8,-1){\makebox(0,0){$V_8$}}
      \put(9,-1){\makebox(0,0){$V_9$}}
      \put(10,-1){\makebox(0,0){$V_{10}$}}
      \multiput(1,-0.25)(1,0){10}{\vector(0,-1){0.5}}
      \put(1,-2){\makebox(0,0){$W_1$}}
      \put(2,-2){\makebox(0,0){$W_2$}}
      \put(3,-2){\makebox(0,0){$W_3$}}
      \put(4,-2){\makebox(0,0){$W_4$}}
      \put(5,-2){\makebox(0,0){$W_5$}}
      \put(6,-2){\makebox(0,0){$W_6$}}
      \put(7,-2){\makebox(0,0){$W_7$}}
      \put(8,-2){\makebox(0,0){$W_8$}}
      \put(9,-2){\makebox(0,0){$W_9$}}
      \put(10,-2){\makebox(0,0){$W_{10}$}}
      \multiput(1,-1.25)(1,0){10}{\vector(0,-1){0.5}}
      \put(12,0){\makebox(0,0){$1_{\hphantom{0}}$}}
      \put(12,-1){\makebox(0,0){$Y_1$}}
      \put(12,-2){\makebox(0,0){$Y_2$}}
      \multiput(12,-0.2)(0,-1){2}{\vector(0,-1){0.5}}
    \end{picture}
  \end{center}
  \caption{ 
    \textbf{(A)} Graphical structure of the aster model for the simulated data 
    in Example 1. The top layer corresponds to survival; these random 
    variables are Bernoulli. The middle layer corresponds to whether or not an 
    individual reproduced; these random variables are also Bernoulli. The 
    bottom layer corresponds to offspring count; these random variables are 
    zero-truncated Poisson. 
    \textbf{(B)} Graphical structure of the aster model for the data in 
    Example 2. The first arrow corresponds to survival which is a Bernoulli 
    random variable. The second arrow corresponds to reproduction count 
    conditional on survival which is a zero-truncated Poisson random variable.
  }
  \label{Fig:graphs}
\end{figure*}

\section{Examples}
\label{s:examples}

\subsection{Simulated Data}
A population of 3000 organisms was simulated to form the dataset used in this 
aster analysis. These data were generated according to the graphical 
structure appearing in panel A of Figure~\ref{Fig:graphs}. There are two 
covariates $(z_1,z_2)$ associated with Darwinian fitness and the aster model 
selected by the likelihood ratio test is a full quadratic model with respect 
to these covariates. 

We partition $\tau$ into $(\gamma^T, \upsilon^T)^T$ where $\gamma \in \R^4$ 
are nuisance parameters and $\upsilon \in \R^5$ are relevant to the estimation 
of expected Darwinian fitness. Here, $\upsilon \in \R^5$ because our model is 
full quadratic in $z_1$ and $z_2$. In this example, the true reducing subspace 
is the space spanned by the first and fourth eigenvectors of the covariance 
matrix of the parameters of interest estimated from the original data. 
We begin by considering envelope estimators constructed using the 1D algorithm. 
The Akaike information criterion and the Bayesian information criterion 
both select $u = 5$. 
Envelope methods are not interesting in this case. 

We now consider envelope estimators constructed from reducing subspaces. 
In the original sample, 
the Bayesian information criterion 
selects the reducing subspace that is the sum of the first, fourth, 
and fifth eigenspaces of $\widehat{\Sigma}_{\upsilon,\upsilon}$ numbered in 
order of decreasing eigenvalues. 
This suggests that the dimension of the envelope space is $u = 3$. The 1D 
algorithm and the reducing subspace approach are in disagreement, consistency 
of model selection is not helpful in this application. We turn to the double 
parametric bootstrap to estimate $\genv$ and the asymptotic variability of 
$\genv$. The Bayesian information criterion is used to select $\G$ at every 
iteration of the first level of the bootstrap. 

The results are seen in Table~\ref{Tab5:target.env}. 
Table~\ref{Tab5:target.env} shows seven individuals that had high values 
of estimated expected Darwinian fitness. Each individual has a unique set of 
traits. The first two columns display the envelope estimator of 
expected Darwinian fitness $\genv$ and its bootstrapped standard error. 
The maximum likelihood estimator of expected Darwinian fitness and its 
bootstrapped standard error are displayed in the third and fourth columns 
respectively. The ratios of bootstrapped standard errors for 
$\hat{g}_{\text{MLE}}$ to $\hat{g}_{\text{env}}$ are displayed in the final 
column. We can see that all of the ratios are greater than 1 which indicates 
that the envelope estimator of expected Darwinian fitness is less variable 
than the maximum likelihood estimator.

\begin{table}
\vspace*{-6pt}
\centering
\def\~{\hphantom{0}}
\caption{Comparison of the maximum likelihood estimator and the 
  envelope estimator for seven individuals with high estimates of 
  expected Darwinian fitness in Example 1. The fifth column is 
  the ratio of bootstrapped standard errors for 
  $\hat{g}_{\text{MLE}}$ to $\hat{g}_{\text{env}}$.}
\label{Tab5:target.env}
  \begin{tabular*}{\columnwidth}{c@{\extracolsep{\fill}}c@{\extracolsep{\fill}}c@{\extracolsep{\fill}}c@{\extracolsep{\fill}}c@{\extracolsep{\fill}}c@{}}
 $\hat{g}_{\text{env}}$ & se$\left(\hat{g}_{\text{env}}\right)$ & $\hat{g}_{\text{MLE}}$ & se$\left(\hat{g}_{\text{MLE}}\right)$ & ratio \\
8.556 & 0.174 & 8.701 & 0.260 & 1.491 \\
9.014 & 0.111 & 8.939 & 0.135 & 1.222 \\
7.817 & 0.414 & 8.054 & 0.442 & 1.069 \\
9.174 & 0.163 & 9.193 & 0.170 & 1.045 \\
9.018 & 0.113 & 9.120 & 0.128 & 1.133 \\
8.612 & 0.162 & 8.518 & 0.278 & 1.709 \\
7.761 & 0.215 & 8.096 & 0.331 & 1.534 \\
\end{tabular*}\vskip18pt
\end{table}

Contour plots of the ratios of estimated standard errors are displayed in the 
technical report \citep{ecktech}. These contour plots show that the envelope 
estimator of expected Darwinian fitness is less variable than the maximum 
likelihood estimator for the majority of the observed data. The region where 
the envelope estimator is less variable includes the values of $z_1$ and $z_2$ 
that maximize estimated expected Darwinian fitness. Variance reduction is also 
obtained when we use the reducing subspace suggested by the Akaike information 
criterion. This is also shown in \citet{ecktech}.

\subsection{\emph{M. guttatus} aster analysis}
The yellow monkeyflower \emph{M. guttatus} has been and is currently a well 
studied flower \citep{lowry, hall, ritland, allen, vickery}. 
\emph{M. guttatus} is a species which comprises many morphologically variable 
populations growing in moist places such as stream banks, meadows and springs 
over a range that extends from the Aleutian Islands to Mexico and from the 
California coast to the Rocky Mountains \citep{vickery}. The lifecycle of the 
individual \emph{M. guttatus} flowers, for our life history analysis, is 
depicted in panel B of Figure~\ref{Fig:graphs}. 
\citet{lowry} performed a life history analysis of \emph{M. guttatus} using 
aster models. One of their interests was to determine which levels of 
genetic background, field site, inversion orientation, and ecotype of the 
flower are associated with high Darwinian fitness. We show that the set of 
candidate trait values thought to maximize expected Darwinian fitness is 
smaller when envelope methodology is incorporated.

\citet{lowry} collected measurements on 2313 \emph{M. guttatus}.
We fit a linear fitness landscape to this data. 
The parameters $\upsilon \in \R^6$ are relevant to the estimation of expected
Darwinian fitness. In the original sample, the Bayesian information criterion 
leads to a selection of a reducing subspace that is the sum of all eigenspaces 
of $\widehat{\Sigma}_{\upsilon,\upsilon}$ with the exception of the fourth 
and fifth eigenspaces. The parametric double bootstrap procedure outlined in 
Sections 3 and 4 is used to estimate the variability of $\genv$ where the 
Bayesian information criterion is used to select $\G$ at every iteration of 
the first level of the bootstrap. 

Table~\ref{Tab6:target.env} shows the results for seven individuals that have 
high values of estimated expected Darwinian fitness through maximum 
likelihood and envelope estimation. Each individual has a unique set of 
traits. We see that both methods agree on the trait values that are expected 
to maximize expected Darwinian fitness. 
We also see that all of the ratios are greater than 1. 
More importantly, this variance reduction implies more precise inference 
in this life history analysis. For example, the envelope estimator can 
statistically distinguish ($\alpha = 0.05$, unadjusted for multiple 
comparisons) the second row of Table~\ref{Tab6:target.env} from the 
fifth row of Table~\ref{Tab6:target.env}. 
The combination of envelope methodology into the aster model framework 
allowed for us to consider a smaller set of traits associated with high 
Darwinian fitness.

\begin{table}
 \vspace*{-6pt}
 \centering
 \def\~{\hphantom{0}}
  \caption{Comparison of the maximum likelihood estimator and the 
    envelope estimator for seven individuals with high estimates of 
    expected Darwinian fitness in Example 2. The fifth column is 
    the ratio of bootstrapped standard errors for 
    $\hat{g}_{\text{MLE}}$ to $\hat{g}_{\text{env}}$.}
\label{Tab6:target.env}
\begin{tabular*}{\columnwidth}{c@{\extracolsep{\fill}}c@{\extracolsep{\fill}}c@{\extracolsep{\fill}}c@{\extracolsep{\fill}}c@{\extracolsep{\fill}}c@{}}
 $\hat{g}_{\text{env}}$ & se$\left(\hat{g}_{\text{env}}\right)$ & $\hat{g}_{\text{MLE}}$ & se$\left(\hat{g}_{\text{MLE}}\right)$ & ratio \\
 9.646 & 0.326 &  9.171 & 0.642 & 1.973 \\
 8.640 & 0.300 &  8.887 & 0.369 & 1.230 \\
 7.659 & 0.315 &  7.603 & 0.361 & 1.144 \\
 7.517 & 0.539 &  7.010 & 0.649 & 1.205 \\
10.943 & 0.607 & 10.475 & 0.896 & 1.476 \\
 7.329 & 0.707 &  6.618 & 1.038 & 1.469 \\
 7.498 & 0.521 &  7.522 & 0.658 & 1.263 \\
\end{tabular*}\vskip18pt
\end{table}


\section{Software} 
\label{s:software}
This paper is accompanied by an R package \texttt{envlpaster}
\citep{envlp-package}, which requires the two R packages for aster models: 
\texttt{aster} \citep{aster-package} and \texttt{aster2} 
\citep{aster2-package}, and also a technical report \citep{ecktech} that 
reproduces the examples in this paper and shows how functions in the 
\texttt{envlpaster} package are used.

\section{Discussion}
\label{s:discussion}

One could think to perform 
envelope methodology with respect to the regression coefficients $\beta$ 
instead of $\tau$. However $\beta$ is not well-defined, one can shift $\beta$ 
with an arbitrarily chosen offset vector without changing the value of the 
mean-value parameters $\tau$ and $\mu$. In addition, whenever we have 
categorical predictors, R software automatically drops one category when it 
has an intercept in the formula, but which category it drops is arbitrary and 
changing which is dropped changes $\beta$.  Envelope methodology is not invariant 
to this form of arbitrary shifting. 
To the best of our knowledge, the exponential family regression applications 
in the envelope model literature exclusively seek inference about $\beta$ 
because this well-definedness issue is not a problem in those applications 
\citep{cook, su, su-inner, cook-scale, cook-pls, 
  cook2, cook-scale-pls, su-group, eck2}. 
The applications of envelope methods to aster models is therefore outside of 
the scope of previous applications to exponential family models. Additionally, 
the methods in this paper can be extended to functions of parameters in 
generalized linear regression models. Aster model are a generalization 
of generalized linear regression models \citep{shaw2}. 



The consequences of potential model selection errors served as the motivation 
for the implementation of the bootstrap procedure in \cite{eck2}. 
In that article, inferences are only given for 
a canonical parameter vector in a multivariate linear regression model. 
We applied the \cite{efron} bootstrap procedures to alleviate 
possible model selection concerns. However, this particular choice of a 
bootstrap procedure is not without flaws. \cite{hjort} mentions that Efron 
does not derive the asymptotic distribution of the final estimator. 
The literature has not reached a consensus 
on the appropriate bootstrap procedure to be implemented when bootstrapping 
depends on data-driven model selection. 
As the literature currently stands, \cite{efron} provides a reasonable 
solution to the problem of potential model selection errors in the application 
of envelope methodology to aster models. The parametric bootstrap does not 
rely on asymptotic normality since it simulates the exact sampling 
distribution of the estimator for some parameter value, and the double 
bootstrap simulates the exact sampling distribution of the estimator for a 
long list of parameter values \citep{geyerLLN}.

Our new envelope estimator does not involve any non-convex optimization 
routines that are both sensitive to starting values and have potential 
problems with local minima. These computational problems can be detrimental 
to the performance of the 1D algorithm. The underlying theory of the 1D 
algorithm justifies the consistency properties of our new envelope estimator. 
In envelope modelling problems with a small number of 
parameters of interest the envelope estimator constructed directly from 
reducing subspaces is preferred since it possesses the same strengths as 
the 1D algorithm without its potential numerical pitfalls. However our 
estimator is currently expensive to compute in moderate $p$ problems. 
In aster analyses $p$ is typically small. 

In many life history analyses, specific trait values which are estimated to 
produce the highest expected Darwinian fitness are of interest. It is common 
practice to only report such trait values \citep{shaw, eck}. Such reporting 
ignores the variability associated with the estimation of expected Darwinian 
fitness. There are likely many trait values having estimated expected 
Darwinian fitness that is statistically indistinguishable from the reported 
values. Our methodology addresses this concern directly. The potential set of 
candidate traits associated with high values of expected Darwinian fitness is 
smaller when the combination of envelope methodology into the aster modelling 
framework is utilized as seen in \citet{ecktech}.
Researchers using our methods will have the potential to make stronger 
inferences about expected Darwinian fitness through our variance reduction 
techniques.

\section*{Acknowledgements}
Daniel J. Eck's research is supported by NIH/NIHCD grant 1DP2HD091799-01. 
We would like to thank David B. Lowry for providing the dataset used in 
Example 2, Xin Zhang for the code that implements the 1D algorithm, and 
Forrest W. Crawford for helpful discussion that led to the strengthening of 
this paper. We would also like to especially thank Amber Eule-Nashoba for 
helpful comments on the technical report.

\bibliographystyle{chicago}

\bibliography{asterenvsources}

\end{document}